\documentclass{article}
\usepackage{fullpage}
\usepackage{authblk}

\usepackage[dvipdfmx]{graphicx} 
\usepackage{amsmath,amsthm,amssymb,mathtools}
\usepackage[colorlinks,citecolor=blue,bookmarks=true,linktocpage]{hyperref}
\usepackage[capitalize]{cleveref}
\usepackage{xcolor}
\usepackage{enumitem}
\setlist[itemize]{itemsep=3pt,parsep=0pt,topsep=2pt}
\usepackage[backend=biber,style=lncs,url=false,arxiv=abs,maxbibnames=100,isbn=false,sortcites=true,style=trad-alpha]{biblatex}
\usepackage{todonotes}
\usepackage[ruled,linesnumbered,vlined]{algorithm2e}
\usepackage{tcolorbox}
\usepackage{xcolor}
\usepackage{physics}
\usepackage{tikz}
\usepackage{subcaption}
\usepackage{diagbox}
\usepackage {tablefootnote}
\usepackage{tabularray}
\usetikzlibrary{positioning, arrows}

\usepackage{lineno}
\newcommand*\patchAmsMathEnvironmentForLineno[1]{
  \expandafter\let\csname old#1\expandafter\endcsname\csname #1\endcsname
  \expandafter\let\csname oldend#1\expandafter\endcsname\csname end#1\endcsname
  \renewenvironment{#1}
     {\linenomath\csname old#1\endcsname}
     {\csname oldend#1\endcsname\endlinenomath}}
\newcommand*\patchBothAmsMathEnvironmentsForLineno[1]{
  \patchAmsMathEnvironmentForLineno{#1}
  \patchAmsMathEnvironmentForLineno{#1*}}
\AtBeginDocument{
\patchBothAmsMathEnvironmentsForLineno{equation}
\patchBothAmsMathEnvironmentsForLineno{align}
\patchBothAmsMathEnvironmentsForLineno{flalign}
\patchBothAmsMathEnvironmentsForLineno{alignat}
\patchBothAmsMathEnvironmentsForLineno{gather}
\patchBothAmsMathEnvironmentsForLineno{multline}
}

\theoremstyle{definition}

\newtheorem{theorem}{Theorem}
\newtheorem{lemma}{Lemma}
\newtheorem{proposition}{Proposition}
\newtheorem{observation}{Observation}
\newtheorem{corollary}{Corollary}

\newtheorem*{claim}{Claim}
\crefname{observation}{Observation}{Observations}
\crefname{algocf}{Algorithm}{Algorithms}

\newenvironment{claimproof}[1]{\par\noindent\textit{Proof of Claim.}\space#1}{\hfill $\blacksquare$}

\graphicspath{{figures/}}

\usepackage{macros}

\title{Forcing a unique minimum spanning tree and a unique shortest path}
\author[1]{Tatsuya Gima}
\author[2]{Andreas Grigorjew}
\author[1]{Yasuaki Kobayashi}
\author[3]{Michael Lampis}
\author[3]{Yiren Lu}
\author[4]{Valia Mitsou}
\author[3]{Edouard Nemery}
\author[5,6]{Yuto Okada}
\author[5]{Yota Otachi}
\author[1]{Takumi Sato}
\affil[1]{Hokkaido University, Japan}
\affil[2]{Institute of Informatics, University of Warsaw, Poland}
\affil[3]{LAMSADE, Universit\'e Paris-Dauphine, PSL University, CNRS UMR7243, France}
\affil[4]{Universit\'e Paris Cit\'e, IRIF, CNRS, France}
\affil[5]{Nagoya University, Japan}
\affil[6]{Research Fellow of Japan Society for the Promotion of Science, Japan}
\date{}
\begin{document}
\maketitle

\begin{abstract}
A \emph{forcing set} $S$ in a combinatorial problem is a set of elements such that there is a unique solution that contains all the elements in $S$.
An \emph{anti-forcing set} is the symmetric concept: a set $S$ of elements is called an anti-forcing set if there is a unique solution disjoint from $S$.
There are extensive studies on the computational complexity of finding a minimum forcing set in various combinatorial problems, and the known results indicate that many problems would be harder than their classical counterparts: the decision version of finding a minimum forcing set for perfect matchings is \NP-complete [Adams et al., Discret. Math. 2004] and that of finding a minimum forcing set for satisfying assignments for 3CNF formulas is $\mathrm{\Sigma}_2^\cP$-complete [Hatami-Maserrat, DAM 2005].
In this paper, we investigate the complexity of the problems of finding minimum forcing and anti-forcing sets for the shortest $s$-$t$ path problem and the minimum weight spanning tree problem.
We show that, unlike the aforementioned results, these problems are tractable, with the exception of the decision version of finding a minimum anti-forcing set for shortest $s$-$t$ paths, which is \NP-complete.
To complement this intractability, we design FPT algorithms for finding a minimum anti-forcing set for shortest $s$-$t$ paths.
\end{abstract}

\section{Introduction}\label{sec:introduction}
This paper focuses on the problem of uniquely determining a solution by prescribing a subset of the solution (or of its complement).
This concept has been extensively studied for various combinatorial problems, under several different names, such as (anti-)forcing number~\cite{AdamsMM04,AfshaniHM04,HararyKZ91,ZhangHLZ25,HararySV07}, critical sets~\cite{GhandehariHM05}, and defining sets~\cite{HatamiM05}.

Harary, Klein, and \v{Z}ivkovi\'c~\cite{HararyKZ91} formulated the \emph{forcing number} (for perfect matchings) of a graph, which is the smallest cardinality of an edge subset such that there is exactly one perfect matching including it.
Prior to this work, the concept has been studied in chemistry as the innate degree of freedom of a Kekul\'e structure~\cite{KleinR87} and gained attention through several studies~\cite{AdamsMM04,AfshaniHM04,ZhangHLZ25}, including a recent survey article \cite{ZhangHLZ25} on the forcing number and on its ``dual'' notion, the anti-forcing number~\cite{DengZ17}.
From the computational perspective, deciding whether the forcing number of an input graph is at most $k$ is known to be \NP-complete even on bipartite graphs~\cite{AfshaniHM04}.
Similarly, given a bipartite graph $G$ and its perfect matching $M$, it is \NP-complete to decide whether there is a subset $S \subseteq M$ of size at most $k$ such that every perfect matching $M'$ of $G$ other than $M$ does not include $S$ entirely (i.e., $S \setminus M' \neq \varnothing$)~\cite{AdamsMM04}.
Deng and Zhang~\cite{DengZ17} proved that, given a bipartite graph $G$ and its perfect matching, it is \NP-complete to decide whether there is a subset $S \subseteq E(G) \setminus M$ of size at most $k$ such that $M$ is the unique perfect matching in $G - S$.
In a more general context, this problem can be viewed as finding a smallest partial solution that is uniquely extended to a (complete) solution.
Viewed from this broader perspective, the computational complexity of such problems has been investigated for various combinatorial problems, such as graph colorings~\cite{HatamiM05}, satisfiability~\cite{HatamiM05,DemaineMSWA16}, puzzles~\cite{DemaineMSWA16,KimuraKF18}, and the minimum vertex cover problem~\cite{HoriyamaKOSS24,HoriyamaSSS25,AnCCKL0S25}.
\cref{tbl:summary} shows several known results.
An important observation is that the computational complexity of (most of) these problems would be harder than that of their classical setting.

\begin{table}[t]\caption{A summary of known and our results.\tablefootnote{Since the solutions of \textsc{3SAT}, \textsc{2SAT}, and \textsc{Vertex Coloring} do not form subsets, these results cannot be directly regarded as results in the forcing model.
However, since these problems deal with unique extensions of partial assignments on subsets of variables or vertices, it is still possible to interpret them as results in the forcing model. (See the discussion in \cite{HatamiM05}.)}}
\label{tbl:summary}
\centering
{\renewcommand{\arraystretch}{1.2}
\begin{tabular}{c|c|c|c|}
\cline{2-4}
\multicolumn{1}{l|}{}   & \multicolumn{1}{c|}{classical} & \multicolumn{1}{c|}{forcing} & \multicolumn{1}{c|}{anti-forcing} \\ \cline{1-4}
\multicolumn{1}{|c|}{\textsc{Perfect Matching}} & \cP            & \NP-complete \cite{AfshaniHM04} & unknown \\ \cline{1-4}
\multicolumn{1}{|c|}{\textsc{3SAT}}             & \NP-complete& $\mathrm{\Sigma}^\cP_2$-complete \cite{HatamiM05} &  \textemdash\\ \cline{1-4}
\multicolumn{1}{|c|}{\textsc{2SAT}}             & \cP            & \NP-complete \cite{DemaineMSWA16} &                         \textemdash \\ \cline{1-4}
\multicolumn{1}{|c|}{\textsc{Vertex Coloring}}         & \NP-complete & $\mathrm{\Sigma}^\cP_2$-complete \cite{HatamiM05}&     \textemdash     \\ \cline{1-4}
\multicolumn{1}{|c|}{\textsc{Vertex Cover}}     & \NP-complete & \multicolumn{2}{c|}{$\mathrm{\Sigma}^\cP_2$-complete \cite{HoriyamaKOSS24}}  \\ \cline{1-4}
\multicolumn{1}{|c|}{\textsc{Bipartite Vertex Cover}} & \cP      & \multicolumn{2}{c|}{\NP-complete \cite{HoriyamaKOSS24}}                  \\ \cline{1-4}
\multicolumn{1}{|c|}{\textsc{Shortest $s$-$t$ Path}}         & \cP & \cP~(Thm. \ref{thm:inclusion-forcing-set-stpath}) &      \NP-complete~(Thm. \ref{thm:exclude-nph}) \\ \cline{1-4}
\multicolumn{1}{|c|}{\textsc{Minimum Spanning Tree}}         & \cP & \multicolumn{2}{c|}{\cP~(Thm. \ref{thm:mst:main})} \\ \cline{1-4}
\end{tabular}
}
\end{table}

We study the computational complexity of the problems of finding a minimum (anti-)forcing set for well-known tractable combinatorial problems, namely the shortest $s$-$t$ path problem and the minimum-weight spanning tree problem.
More specifically, let $G$ be a directed graph with $s, t \in V(G)$ and let $w\colon E(G) \to \mathbb R_{+}$ be an edge-weight function.
We say that $S \subseteq E(G)$ is a \emph{forcing set for shortest $s$-$t$ paths} if there is exactly one shortest path $P$ from $s$ to $t$ in $(G, w)$ such that $S \subseteq E(P)$.
Similarly, $S \subseteq E(G)$ is an \emph{anti-forcing set for shortest $s$-$t$ paths} if there is exactly one shortest path $P$ from $s$ to $t$ in $(G, w)$ such that $S \cap E(P) = \varnothing$.
Our first two problems are stated formally as follows.
\begin{tcolorbox}
\begin{description}
  \setlength{\itemsep}{0pt}
  \item[Problem:] {\shortincsssp} / {\shortexcsssp}
  \item[Input:] A weighted digraph $G$ with $w\colon E(G) \to \mathbb R_{+}$, $s,t \in V(G)$, and $k \in \mathbb N$.
  \item[Goal:] Decide whether there is a forcing set / anti-forcing set $S \subseteq E(G)$ with $|S| \le k$ for shortest $s$-$t$ paths in $(G, w)$.
\end{description}
\end{tcolorbox}


Similarly, we can define forcing and anti-forcing sets for minimum weight spanning trees.
They are defined as follows.
\begin{tcolorbox}
\begin{description}
  \setlength{\itemsep}{0pt}
  \item[Problem:] {\shortincmst} / {\shortexcmst}
  \item[Input:] A weighted graph $G$ with $w\colon E(G) \to \mathbb R$.
  \item[Goal:] Decide whether there is a forcing set / anti-forcing set $S \subseteq E(G)$ with $|S| \le k$ for minimum weight spanning trees in $(G, w)$.
\end{description}
\end{tcolorbox}

Let us note that the weight of an edge can be negative in the latter two problems, while the former two problems require the weight to be positive.

It is not hard to see that these four problems belong to \NP, as we can check in polynomial time, given an edge set $S$, whether there is a unique optimal solution including or avoiding it.
This motivates us to investigate whether these problems are \NP-complete or polynomial-time solvable.

For {\shortincmst} and {\shortexcmst}, we can naturally generalize these problems to the settings in matroids.
This generalization allows us to handle these two problems in a uniform way.
In particular, we observe that the problem of finding a minimum forcing set for minimum weight bases in a matroid $M$ is equivalent to that of finding a minimum weight anti-forcing set for maximum weight bases in the dual matroid $M^*$.
We design a polynomial-time algorithm for computing a minimum anti-forcing set for maximum weight bases in a matroid $M$, assuming that $M$ is given as an independence oracle (see \cref{sec:preliminaries} for details).
This yields polynomial-time algorithms for {\shortincmst} and {\shortexcmst} (\cref{thm:mst:main}).
Our polynomial-time algorithm exploits a well-known greedy algorithm for the maximum weight basis problem.

For {\shortincsssp} and {\shortexcsssp}, we can observe an intriguing distinction between them: {\shortincsssp} is polynomial-time solvable (\cref{thm:inclusion-forcing-set-stpath}), while {\shortexcsssp} is \NP-complete (\cref{thm:exclude-nph}).
To the best of our knowledge, this is the first example in which there is a complexity gap between the two models (see \cref{tbl:summary}).
Our polynomial-time algorithm for {\shortincsssp} is based on the simple reduction that reduces the problem on general weighted directed graphs to that on unweighted directed acyclic graphs.
This reduction enables us to design a dynamic programming algorithm.
For {\shortexcsssp}, we give a polynomial-time reduction from {\mcbiclique}, proving that it is W[1]-hard when parameterized by the distance between $s$ and $t$.
To overcome this intractability, we design several FPT algorithms for {\shortexcsssp}.

\section{Preliminaries}\label{sec:preliminaries}
Let $E$ be a finite set.
A \emph{property} is a subset $\mathrm{\Pi} \subseteq 2^E$.
For $X \in \mathrm{\Pi}$, a set $S \subseteq E$ is called a \emph{forcing set for $X$} if it holds that $S \subseteq X$ and $S \not\subseteq X'$ for every $X' \in \mathrm{\Pi} \setminus \{X\}$.
In other words, $S$ is a forcing set of $X$ if $X$ is the unique set in $\mathrm{\Pi}$ that contains $S$.
Similarly, for $X \in \mathrm{\Pi}$, a set $S \subseteq E$ is called an \emph{anti-forcing set for $X$} if it holds that $S \cap X = \varnothing$ and $S \cap X' \neq \varnothing$ for every $X' \in \mathrm{\Pi} \setminus \{X\}$.
A set $S \subseteq E$ is called a \emph{forcing set} (resp.~\emph{anti-forcing set}) \emph{for property $\mathrm{\Pi} \subseteq 2^E$} if $S$ is a forcing set (resp.~anti-forcing set) for some $X \in \mathrm{\Pi}$.
We sometimes omit the target property $\mathrm{\Pi}$ when it is clear from the context. 

For a (directed) graph $G$, we denote its vertex set by $V(G)$ and its edge set by $E(G)$.
For an edge set $X \subseteq E(G)$, $G - X$ denotes the graph obtained from $G$ by deleting all edges in $X$.
Let $G$ be a directed graph.
For a vertex $v \in V(G)$, $N^-(v)$ denotes the set of in-neighbors $\qty{u : (u,v) \in E(G)}$ of $v$.
For a directed edge $e = (u,v) \in E(G)$, $\head(e)$ and $\tail(e)$ denote $v$ and $u$, respectively.

\paragraph{Matroids}
Let $E$ be a finite set and let $\mathcal I \subseteq 2^E$.
A pair $M = (E, \mathcal I)$ is called a \emph{matroid} if $\mathcal I$ is nonempty; for $X \in \mathcal I$, every subset of $X$ belongs to $\mathcal I$; and for $X, Y \in \mathcal I$ with $|Y| > |X|$, $Y \setminus X$ has an element $e$ such that $X \cup \{e\} \in \mathcal I$.
Each set in $\mathcal I$ (resp.~not in $\mathcal I$) is said to be \emph{independent} (resp.~\emph{dependent}) in $M$.
A maximal independent set is called a \emph{basis} and a minimal dependent set is called a \emph{circuit} in $M$.
Every matroid has the following \emph{symmetric basis-exchange property}~\cite{Brualdi_1969}: for any two bases $A$ and $B$ and for $a\in A\setminus B$, there is $b \in B \setminus A$ such that both $(A\cup\{b\}) \setminus\{a\}$ and $(B \cup\{a\}) \setminus \{b\}$ are bases.
The collection of all bases of $M$ is denoted by $\mathcal B(M)$.
A \emph{loop} in $M$ is an element $e \in E$ such that $\{e\}$ is dependent.

\begin{proposition}[Corollary 1.2.6 in~\cite{Oxley92}] \label{prop:circuit-basis}
    Let $B \in \mathcal B(M)$ and $e \notin B$.
    Then, $B \cup \{e\}$ contains a unique circuit $C$ of $M$.
    Moreover, we have $(B \cup \{e\}) \setminus \{e'\} \in \mathcal B(M)$ for $e' \in C$.
\end{proposition}

Let $M = (E, \mathcal I)$ be a matroid.
It is known that the pair $M^* \coloneqq (E, \mathcal I^*)$ with
\begin{align*}
    \mathcal{I}^* = \{X \subseteq E: X \subseteq E \setminus B \text{ for some } B \in \mathcal B(M) \}
\end{align*}
forms a matroid~\cite{Oxley92}, which is called the \emph{dual matroid} of $M$.
Note that $(M^*)^* = M$ for any matroid $M$.
For $X \subseteq E$, we define pairs
\begin{align*}
    M \mid X &= (X, \{I \subseteq X : I \in \mathcal I\}),\\
    M \contract X &= (E \setminus X, \{I \subseteq E \setminus X : \exists B \in \mathcal B(M \mid X) \text{ s.t. } B \cup I \in \mathcal I\}).
\end{align*}
These pairs are called the \emph{restriction} and the \emph{contraction} of $M$ with respect to $X$, respectively.
It is known that these pairs are also matroids~\cite{Oxley92} for any $X \subseteq E$.
In this paper, a matroid $M$ is given as an \emph{independence oracle}, that is, we can query the oracle to decide whether a subset $X \subseteq E$ is independent in $M$.

\begin{proposition}[Proposition 3.1.10 in \cite{Oxley92}]\label{prop:circuits-contraction}
    Let $X \subseteq E$.
    For every circuit $C$ of $M$ with $C \not\subseteq X$, there is a circuit $C^*$ of $M \contract X$ such that $C^* \subseteq C \setminus X$.
\end{proposition}

Let $G = (V, E)$ be a multigraph.
A subset of edges is said to be \emph{acyclic} if it induces a forest in $G$.
Let $\mathcal F \subseteq 2^E$ be the collection of all acyclic edge subsets of $G$.
Then, it is known that the pair $M = (E, \mathcal F)$ forms a matroid, which is called a \emph{graphic matroid}~\cite{Oxley92}.


\section{Forcing a unique minimum weight matroid basis}\label{sec:unique_mst}

In this section, we give polynomial-time algorithms for \shortincmst{} and \shortexcmst{}.
The basic strategy of both algorithms follows Kruskal's algorithm~\cite{Kruskal56} for computing a minimum weight spanning tree.
We first briefly explain our algorithms for \shortincmst{} and \shortexcmst{} in \cref{subsec:example-mst} and then generalize them to the setting of matroids in \cref{subsec:matroid}.

\subsection{Special case: Minimum weight spanning trees}\label{subsec:example-mst}

\begin{algorithm}[tb]
    \caption{An algorithm for \shortexcmst.}
    \label{alg:anti-spanning-tree}
    \KwIn{A multigraph $G$ with weight function $w: E(G) \to \mathbb R$.}
    \KwOut{A minimum anti-forcing set for minimum weight spanning trees of $(G,w)$.}
    \Begin{
        $S \leftarrow \varnothing$; $T \leftarrow \varnothing$\;
        \While {$E(G) \neq \varnothing$}{
            $w_{\min} \leftarrow \min_{e \in E(G)} w(e)$\;
            $E_{\min} \leftarrow \{e \in E(G) : w(e) = w_{\min}\}$\;
            Let $G_{\min}$ be the subgraph $(V(G), E_{\min})$ of $G$\;
            Let $F_{\min}$ be a maximal forest of $G_{\min}$\;
            $T \leftarrow T \cup F_{\min}$\;
            $S \leftarrow S \cup (E_{\min} \setminus (F_{\min} \cup \{e \in E_{\min} : e \text{ is a self-loop in } G_{\min}\}))$\;
            Let $G$ be the graph obtained by contracting all edges in $E_{\min}$\;
        }
        \Return $S$\;
    }
\end{algorithm}

\begin{algorithm}
    \caption{An algorithm for \shortincmst.}
    \label{alg:spanning-tree}
    \KwIn{A multigraph $G$ with weight function $w\colon E(G) \to \mathbb R$.}
    \KwOut{A minimum forcing set for minimum weight spanning trees of $(G,w)$.}
    \Begin{
        $S \leftarrow \varnothing$; $T \leftarrow \varnothing$\;
        \While {$E(G) \neq \varnothing$}{
            $w_{\min} \leftarrow \min_{e \in E(G)} w(e)$\;
            $E_{\min} \leftarrow \{e \in E(G) : w(e) = w_{\min}\}$\;
            Let $G_{\min}$ be the subgraph $(V, E_{\min})$ of $G$\;
            Let $F_{\min}$ be a maximal forest of $G_{\min}$\;
            $T \leftarrow T \cup F_{\min}$\;
            $S \leftarrow S \cup \{e \in F_{\min} : e \text{ is not a bridge of $G_{\min}$}\}$\;
            Let $G$ be the graph obtained by contracting all edges in $E_{\min}$\;
        }
        \Return $S$\;
    }
\end{algorithm}

First, we sketch an intuition behind our algorithm for {\shortexcmst}. 
A formal discussion, including the correctness of the algorithm, is deferred to \cref{subsec:matroid}.
Let $G$ be a connected edge-weighted multigraph.
Contracting an edge of $G$ may create parallel edges or self-loops, and we do not remove them in the contraction operation.
Let us recall Kruskal's algorithm for computing a minimum weight spanning tree of $G$:
Starting from $T = \varnothing$, we repeatedly choose a minimum weight edge $e$ of $G$ that does not form a cycle with the previously chosen edges $T$ and add it to $T$ until it becomes a spanning tree of $G$.
Instead of adding edges one by one, we can modify this procedure to add, in bulk, a maximal forest $F$ consisting only of minimum weight edges $E_{\min}$ that have not been chosen yet.
We repeatedly apply this to the graph obtained from $G$ by contracting all edges in $E_{\min}$ as long as $E(G)$ is nonempty.
It is easy to observe that this procedure also computes a minimum weight spanning tree $T$ of $G$.
Since the choice of a maximal forest $F$ consisting of edges in $E_{\min}$ is not unique, the solution $T$ obtained by this procedure is not unique as well.
Thus, to force the solution to be unique, it is necessary to include all edges in $E_{\min} \setminus F$, except for self-loops, as an anti-forcing set.
The algorithm that formalizes this intuition is given in \Cref{alg:anti-spanning-tree}.
For \shortincmst{}, we can design a similar algorithm: Instead of including all non-loop edges in $E_{\min}\setminus F$, include all edges in $F$ that are not bridges in $G_{\min}$.
The pseudocode of this algorithm is given in~\Cref{alg:spanning-tree}.
It can be easily confirmed that both algorithms for \shortexcmst{} and \shortincmst{} run in time $\bigoh(m\log n)$ using a standard analysis of Kruskal's algorithm and a linear-time algorithm for enumerating bridges~\cite{Tarjan74}.
The proof of their correctness is deferred to the next subsection.

\begin{theorem}\label{thm:mst:main}
    \shortincmst{} and \shortexcmst{} can be solved in time $\bigoh(m\log n)$, where $n = |V(G)|$ and $m = |E(G)|$.
\end{theorem}

\subsection{General case: Matroid bases}\label{subsec:matroid}
    In this subsection, we consider generalized versions of \shortincmst{} and \shortexcmst{}, where the goals are to find minimum forcing and anti-forcing sets for minimum weight bases of a matroid $M$.
    Clearly, \shortincmst{} and \shortexcmst{} are special cases, where $M$ is a graphic matroid.
\begin{theorem}\label{thm:polytime-matroid}
    Let $M = (E, \mathcal I)$ be a matroid with weight function $w \colon E \to \mathbb R$.
    Assume that $M$ is given as an independence oracle.
    Then, a minimum forcing set for minimum weight bases of $(M, w)$ can be computed in polynomial time.
    Similarly, a minimum anti-forcing set for minimum weight bases of $(M, w)$ can be computed in polynomial time as well.
\end{theorem}

Let $\mathcal B_{\min}(M)$ and $\mathcal B_{\max}(M)$ be the collections of minimum and maximum weight bases of $M$, respectively.
At first, we show that the problem of finding a minimum forcing set for $\mathcal B_{\min}(M)$ is equivalent to that of finding a minimum anti-forcing set for $\mathcal B_{\max}(M^*)$, where $M^*$ is the dual matroid of $M$.

\begin{observation}\label{obs:matroid-dual}
    A set $S \subseteq E$ is a forcing set for $\mathcal B_{\min}(M)$ if and only if 
    $S$ is an anti-forcing set for $\mathcal B_{\max}(M^*)$.
\end{observation}
\begin{proof}
    Let $B \subseteq E$ be a minimum weight basis of $(M, w)$ and $B^* = E \setminus B$.
    Note that $B^*$ is a basis of the dual matroid $M^*$.
    Let $S\subseteq E$.
    Now, $S \subseteq B$ if and only if $S \setminus B = \varnothing$.
    Combined with $ S \setminus (E\setminus B^*) = S \cap B^*$, we conclude that $S \subseteq B$ if and only if $S \cap B^* = \varnothing$.
    Since $B$ is a minimum weight basis of $(M, w)$ if and only if $B^*$ is a maximum weight basis of $(M^*, w)$, the claim holds.
\end{proof}

\cref{obs:matroid-dual} suggests the relations in \cref{fig:matroid-forcing-set-equivalence}, enabling us to find a minimum forcing set for $\mathcal B_{\min}(M)$ by applying an algorithm for finding a minimum anti-forcing set for $\mathcal B_{\max}(M^*)$.
Note that a basis $B$ minimizes the value $w(B)$ if and only if it maximizes $-w(B)$. 
\begin{figure}[htb]
    \centering
    \begin{tikzpicture}[every node/.style={outer sep=0.2em},]
        \node[draw,align=center, ] (ifsmin) {$S$ is a \textbf{forcing set} for\\\textbf{min.} weight bases of $(M, w)$};
        \node[draw,align=center, below=of ifsmin]    (ifsmax)  {$S$ is a \textbf{forcing set} for\\\textbf{max.} weight bases of $(M, -w)$};
        \node[draw,align=center, right=2.5 of ifsmax] (efsmax)  {$S$ is an \textbf{anti}-forcing set for\\ \textbf{max.} weight bases of $(M^*, w)$};
        \node[draw,align=center, above= of efsmax] (efsmin)  {$S$ is an \textbf{anti}-forcing set for\\ \textbf{min.} weight bases of $(M^*, -w)$};
        \draw[implies-implies, double equal sign distance] (ifsmin) -- (efsmax) node[midway, above=0.1] {$M \leftrightarrow M^*$} node[midway, below=0.12] {\cref{obs:matroid-dual}};
        \draw[implies-implies, double equal sign distance] (ifsmax) -- (efsmin);
        \draw[implies-implies, double equal sign distance] (ifsmin) -- (ifsmax) node[midway, left ] {$w \leftrightarrow -w$};
        \draw[implies-implies, double equal sign distance] (efsmin) -- (efsmax) node[midway, left ] {$w \leftrightarrow -w$};
    \end{tikzpicture}
    \caption{The equivalence of the problems.}
    \label{fig:matroid-forcing-set-equivalence}
\end{figure}
Now, we are ready to describe our algorithm for computing a minimum anti-forcing set for $\mathcal B_{\min}(M)$, which is shown in \Cref{alg:exclude-matroid}.
The underlying idea of the algorithm is analogous to \Cref{alg:anti-spanning-tree}.
At each iteration of the main loop, we consider the matroid $M_{\min}$ consisting only of the minimum weight elements $E_{\min}$, and select a basis $B_{\min}$ from $M_{\min}$.
The crux of its correctness is that it is necessary and sufficient to include all the elements in $E_{\min} \setminus B_{\min}$ except for loops in any anti-forcing set for $\mathcal B_{\min}(M)$.

\begin{algorithm}[tb]
    \caption{An algorithm for a minimum anti-forcing set for $\mathcal B_{\min}(M)$}
    \label{alg:exclude-matroid}
    \KwIn{Matroid $M = (E, \mathcal I)$ with weight function $w\colon E \to \mathbb R$.}
    \KwOut{A minimum anti-forcing set for $\mathcal B_{\min}(M)$.}
    \Begin{
        $S \leftarrow \varnothing$; $B \leftarrow \varnothing$\;
        \While {$E \neq \varnothing$}{
            $w_{\min} \leftarrow \min_{e \in E} w(e)$\;
            $E_{\min} \leftarrow \{e \in E : w(e) = w_{\min}\}$\;
            $M_{\min} \leftarrow M \mid E_{\min}$\;
            Let $B_{\min}$ be a basis of $M_{\min}$\;
            $B \leftarrow B \cup B_{\min}$\;
            $S \leftarrow S \cup (E_{\min} \setminus (B_{\min} \cup \{e \in E_{\min} : e \text{ is a loop in } M_{\min}\}))$\;
            $M \leftarrow M\contract E_{\min}$\;
        }
        \Return $S$\;
    }
\end{algorithm}

\begin{lemma}\label{lem:correctness-exclude-matroid}
\Cref{alg:exclude-matroid} 
returns a minimum anti-forcing set for $\mathcal B_{\min}(M)$.
\end{lemma}
\begin{proof}
    First, we show that the output $S$ of \Cref{alg:exclude-matroid} is an anti-forcing set for $\mathcal B_{\min}(M)$.
    From the correctness of the greedy algorithm for computing a minimum weight basis of a matroid,
    at the end of the main loop, we have $B \in \mathcal B_{\min}(M)$.
    Moreover, we have $B \cap S = \varnothing$.

    Let $B' \in \mathcal B_{\min}(M)$ with $B' \neq B$.
    It suffices to show that $B' \cap S \neq \varnothing$.
    Now, there exists an element $e \in B' \setminus B$ since $B$ and $B'$ are distinct bases of $M$.
    We can assume that, at some iteration of the main loop, $e\in E_{\min}$ and $e \notin B_{\min}$ hold.
    Since $e$ is not a loop in $M_{\min}$ and $e \in E_{\min} \setminus B_{\min}$, we have $e \in S$.
    Hence, $B' \cap S \neq \varnothing$. This concludes that $S$ is an anti-forcing set for $\mathcal B_{\min}(M)$.

    Next, we show that $S$ is a minimum anti-forcing set for $\mathcal B_{\min}(M)$.
    Let $S'$ be a minimum anti-forcing set for $\mathcal B_{\min}(M)$ and $B' \in \mathcal B_{\min}(M)$ be the unique minimum weight basis of $M$ with $B' \cap S' = \varnothing$.
    Let $w_i$ be the value of $w_{\min}$ in the $i$-th iteration of the while-loop and let $E_i = \{e \in E : w(e) = w_i\}$. 
    For each $i$, let $E_{\le i} = E_1 \cup \dots \cup E_i$ and let $M_i = M \contract E_{\le i-1}$.
    Note that $M_i \mid E_i = M_{\min}$ in the $i$-th iteration.
    
    \begin{claim}
        $|B \cap E_i| = |B' \cap E_i|$ holds for all $i$.
    \end{claim}
    \begin{claimproof}
        We prove a slightly stronger statement: for $B_1, B_2 \in \mathcal B_{\min}(M)$, it holds that $|B_1 \cap E_i| = |B_2 \cap E_i|$ for all $i$.
        Suppose otherwise.
        We choose $B_1, B_2 \in \mathcal B_{\min}(M)$ with $|B_1 \cap E_j| \neq |B_2 \cap E_j|$ for some $j$ in such a way that the symmetric difference $|B_1 \symdif B_2| \coloneqq |B_1 \setminus B_2| + |B_2 \setminus B_1|$ is minimized.
        We assume without loss of generality that $|B_1 \cap E_j| > |B_2 \cap E_j|$.
        Since $B_1 \neq B_2$, by the symmetric basis-exchange property, for $e_1 \in (B_1 \setminus B_2) \cap E_j$, there exists $e_2 \in B_2 \setminus B_1$ such that both $B_1' \coloneqq (B_1 \cup \{e_2\}) \setminus \{e_1\}$ and $B'_2 \coloneqq (B_2 \cup \{e_1\}) \setminus \{e_2\}$ are bases in $M$.
        Since $B_1, B_2 \in \mathcal B_{\min}(M)$, we have $w(e_1) = w(e_2)$, and hence $e_1, e_2 \in E_j$.
        Moreover, $B'_1, B'_2 \in \mathcal B_{\min}(M)$.
        This contradicts the choice of $B_1, B_2$ as 
        \begin{align*}
            |B'_1 \symdif B_2| = |(B_1 \cup \{e_2\}) \setminus \{e_1\}) \symdif B_2| = |B_1 \symdif B_2| - 2 < |B_1 \symdif B_2|
        \end{align*}
        and $|B'_1 \cap E_j| = |B_1 \cap E_j| \neq |B_2 \cap E_j|$.
    \end{claimproof}


    Suppose to the contrary that $|S'| < |S|$.
    This implies that there is an index $j$ such that $|S' \cap E_j| < |S \cap E_j|$.
    By the above claim, we have $|B \cap E_j| = |B' \cap E_j|$.
    Let $e \in E_j \setminus (B' \cup S')$ that is not a loop in $M_j \mid E_{j}$.
    We can choose such an element $e$ since $|E_j| = |B \cap E_j| + |S \cap E_j| + |L_j|$, where $L_j$ is the set of loops in $M_j \mid E_j$, and $|S' \cap E_j| < |S \cap E_j|$.
    Note that $e$ is not a loop of $M_j$.
    Due to \cref{prop:circuit-basis}, $B' \cup \{e\}$ has a unique circuit $C$ of $M$.
    If $C$ contains an element $e' \in E \setminus E_{\le j}$, then $(B' \cup \{e\}) \setminus \{e'\}$ is a basis of $M$ with weight strictly smaller than $B'$, contradicting $B' \in \mathcal B_{\min}(M)$.
    Thus, $C$ consists of only elements in $E_{\le j}$.
    If $C$ contains an element $e' \in E_j$, the basis $(B' \cup \{e\}) \setminus \{e'\}$, which belongs to $\mathcal B_{\min}(M)$, avoids $S'$, contradicting the uniqueness of $B'$.
    Hence, all the elements of $C$ except for $e$ belong to $E_{\le j - 1}$.
    As $C \not\subseteq E_{\le j - 1}$, by~\cref{prop:circuits-contraction}, there is a circuit $C^*$ of $M_j$ such that $C^* \subseteq C \setminus E_{\le j-1}$.
    This circuit is indeed a singleton $C^* = \{e\}$, contradicting the fact that $e$ is not a loop in $M_j$.
\end{proof}

Finally, we consider the running time of \Cref{alg:exclude-matroid}.
Given a matroid $M = (E, \mathcal I)$ as an independence oracle and $X, Y \subseteq E$, we can decide whether $Y$ is independent in $M \mid X$, in $M \contract X$, and in $M^*$ with a polynomial number of oracle calls to $M$.
Thus, each step of \Cref{alg:exclude-matroid} can be performed in polynomial time.
Thus, \cref{thm:polytime-matroid} holds.

Similarly to \cref{thm:inclusion-forcing-set-stpath}, we can compute smallest forcing and anti-forcing sets for a given $B^* \in \mathcal B_{\min}(M)$ in polynomial time, by just taking $B_{\min}$ at line~7 in \Cref{alg:exclude-matroid} as $B_{\min} = B^* \cap E_{\min}$.

\begin{corollary}
    Let $M = (E, \mathcal I)$ be a matroid with weight function $w: E \to \mathbb R$.
    Assume that $M$ is given as an independence oracle.
    Given a minimum weight basis $B^* \in \mathcal B_{\min}(M)$, a minimum forcing set / a minimum anti-forcing set for $B^*$ can be computed in polynomial time.
\end{corollary}



\section{Forcing a unique shortest path}\label{sec:unique_stpath}
In this section, we discuss {\shortincsssp} and {\shortexcsssp}.
Recall that in these problems we are given a directed graph $G$ with $s, t \in V(G)$ and an edge-weight function $w\colon E \to \mathbb R_{+}$.
The numbers of vertices and edges in $G$ are denoted by $n$ and $m$, respectively.
A forcing (resp.~anti-forcing) set $S \subseteq E(G)$ (for shortest $s$-$t$ paths) is \emph{minimal} if any proper subset of $S$ is not a forcing (resp.~anti-forcing) set.
It is easy to observe that for any minimal (anti-)forcing set $S$, every edge in $S$ is contained in a shortest $s$-$t$ path in $G$.
This would reduce our problems to the cases where $G$ is acyclic.
More specifically, we let $d_s : V(G) \to \mathbb R_{+} \cup \{0, \infty\}$ be the distance labeling from $s$ in $G$, that is, $d_s(v)$ is the (shortest) distance from $s$ to $v$ in $(G, w)$ for $v \in V(G)$.
We remove all the edges $e = (u, v)$ that do not satisfy $d_s(v) = d_s(u) + w(e)$, which are edges that do not belong to any shortest path from $s$ to $v$.
We also remove all vertices (and their incident edges) that are not reachable from $s$ or not reachable to $t$.
As $w(e) > 0$ for each $e \in E$, the graph $G'$ obtained in this way is indeed acyclic, which can be computed by a standard shortest path algorithm in $\bigoh(m + n\log n)$ time.
The following observation immediately follows from the fact that each shortest $s$-$t$ path in $(G, w)$ is an $s$-$t$ path in $G'$, and vice versa.

\begin{observation}\label{obs:shortest-path-dag}
Let $S \subseteq E$.
Then, $S$ is a minimal forcing set in $G$ if and only if it is a minimal forcing set in $G'$.
Similarly, $S$ is a minimal anti-forcing set in $G$ if and only if it is a minimal anti-forcing set in $G'$.
\end{observation}

This observation offers several advantages: we can assume that the input directed graph $G$ is acyclic and all the paths from $s$ to $t$ are shortest paths in $G$, making the subsequent discussions simple.


%

\subsection{Forcing set}\label{subsec:include-stpath}

In this subsection, we describe an algorithm for computing a minimum forcing set for shortest $s$-$t$ paths.
By~\cref{obs:shortest-path-dag}, we can assume that the given graph $G$ is acyclic.
Moreover, we can ignore the weight of edges, as every path from $s$ to $t$ in $G$ is a shortest $s$-$t$ path in the original graph.

For $v \in V(G)$, let $\mathcal P_{v}$ be the set of all paths from $s$ to $v$ in $G$.
For vertices $u, v \in V(G)$, we write $u \ureach v$ if there is exactly one path from $u$ to $v$ in $G$.
This relation $\ureach$ is reflexive, i.e., $v \ureach v$ holds for any $v\in V(G)$.
We assume that $s$ has out-degree at least~$2$ since otherwise we can contract the (unique) out-going edge from $s$ without affecting the solution.
We define $\OPT[e]$ and $\OPT[v]$ as
\begin{align*}
    \OPT[e] &\coloneqq \min\qty{|S| : S \text{ is a forcing set for } \mathcal P_{v} \text{ with } e \in S}
\end{align*}
for $e = (u, v) \in E(G)$ and 
\begin{align*}
    \OPT[v] \coloneqq \min_{u \in N^{-}(v)} \OPT[(u, v)]
\end{align*}
for $v \in V(G) \setminus \{s\}$, while $\OPT[s] \coloneqq 0$. 
From now on, we show how to compute the values of $\OPT$ and then a minimum forcing set for $\mathcal P_t$ from these values.

The following lemma gives a characterization of a forcing set for a specific path $P \in \mathcal P_v$.
\begin{lemma}\label{lem:sssp:forcing}
    Let $S \subseteq E(G)$ and let $P \in \mathcal P_v$ be a path from $s$ to $v$ in $G$ such that all edges in $S$ are contained in $P$.
    Let $e_1, \dots, e_k$ be the edges in $S$ appearing in this order on $P$.
    For $1 \le i \le k$, we let $e_i = (t_i, s_i)$, and let $s_0 = s$ and $t_{k + 1} = v$.
    Then, $S$ is a forcing set for $P$ if and only if $s_i \ureach t_{i + 1}$ for all $0 \le i \le k$.
\end{lemma}
\begin{proof}
    Suppose that $S$ is a forcing set for $P$.
    If there are at least two paths from $s_i$ to $t_{i+1}$ for some $i$, we can conclude that there is at least one path $P' \in \mathcal P_v$ with $P' \neq P$ that contains all the edges in $S$, which contradicts the uniqueness of $P$.

    Conversely, suppose that $s_i \ureach t_{i + 1}$ for all $0 \le i \le k$.
    If $S$ is not a forcing set for $P$, there is another path $P' \in \mathcal P_v$ such that all the edges in $S$ are contained in $P'$.
    As $P \neq P'$, at least one edge of $P$ is not contained in $P'$.
    We can assume that this edge appears on the subpath $P_i$ of $P$ between $s_i$ and $t_{i + 1}$.
    Since $P'$ also passes through both $s_i$ and $t_{i+1}$, the subpath of $P'$ between them is distinct from $P_i$, contradicting $s_i \ureach t_{i + 1}$.
\end{proof}


\begin{corollary}\label{cor:objective-st-path-dp}
    The minimum size of a forcing set for $\mathcal P_t$ is equal to
    \begin{align*}
        \min_{v \in V(G)}\qty{\OPT[v] : v \ureach t}.
    \end{align*}
\end{corollary}

We turn to a polynomial-time algorithm to compute the values of $\OPT$.
This immediately yields a polynomial-time algorithm for {\shortincsssp} due to \cref{cor:objective-st-path-dp}.
\begin{lemma}\label{lem:include-stpath-dp}
    For every $e=(u,v)\in E(G)$, it holds that
    \begin{align*}
        \OPT[e] = \min_{w \in V(G)}\qty{\OPT[w]: w \ureach u} + 1.
    \end{align*}
\end{lemma}
\begin{proof}
    Let $S \subseteq E(G)$ be a forcing set for $P \in \mathcal P_v$ such that $|S| = \OPT[e]$ and $e \in S$.
    By~\cref{lem:sssp:forcing}, there is a vertex $w'$ in $P$, which is either $s$ or the head of an edge in $S \setminus \{e\}$ with $w' \ureach u$.
    Observe that $S \setminus \{e\}$ is a forcing set for the subpath of $P$ between $s$ and $w'$, as otherwise there are two paths from $s$ to $v$ including all edges in $S$, contradicting the uniqueness of $P$.
    Thus, we have
    \begin{align*}
        |S| =  |S \setminus \{e\}|+ 1 \ge \OPT[w'] + 1 \ge \min_{w \in V(G)}\qty{\OPT[w] : w \ureach u} + 1.
    \end{align*}

    Conversely, let $w$ be a vertex minimizing $\OPT[w]$ under the condition that $w \ureach u$ holds.
    Suppose first that $\OPT[w] = 0$.
    By definition, we have $w = s$.
    Since $s \ureach u$, there is a unique path from $s$ to $u$, which is denoted by $P_{su}$.
    By concatenating $P_{su}$ and $e$ in this order, we have a path $P$ from $s$ to $v$ in $G$.
    Observe that there is exactly one path from $s$ to $v$ passing through $e$ as $s \ureach u$ and every path containing $e$ must pass through $u$.
    Thus, $\{e\}$ is a forcing set for $P$, implying that $\OPT[e] \le \OPT[w] + 1$.

    Suppose otherwise.
    In this case, $\OPT[w] = \OPT[e']$ for some edge $e'$ incoming to $w$.
    Let $S'$ be a forcing set for $P_{sw} \in \mathcal P_w$ such that $|S'| = \OPT[w]$ and $e' \in S'$. 
    Similarly to the above case, the path obtained by concatenating $P_{sw}$, $P_{wu}$, and $e$ in this order is the unique path $P$ from $s$ to $v$ containing all edges in $S' \cup \{e\}$.
    Thus, $S' \cup \{e\}$ is a forcing set for $P$, implying that $\OPT[e] \le \OPT[w] + 1$.
\end{proof}

Now, we describe our dynamic programming algorithm to compute the values of $\OPT$.
We first decide whether $u \ureach v$ holds for each pair of vertices $u, v \in V(G)$.
This can be done in total time $\bigoh(nm)$ for all vertex pairs in $G$.
We can evaluate $\OPT[e]$ and $\OPT[v]$ for each $e \in E(G)$ and each $v \in V(G)$ in time $\bigoh(n)$, assuming that these values are evaluated in a dynamic programming manner.
By \cref{cor:objective-st-path-dp}, we can compute a minimum forcing set for $\mathcal P_t$ in time $\bigoh(nm)$.

It is easy to extend our algorithm to compute a minimum forcing set for a specific path $P$ by only evaluating $\OPT[e]$ and $\OPT[v]$ for each edge $e$ and vertex $v$ on the path $P$.
Therefore, we have the following theorem.

\begin{theorem}\label{thm:inclusion-forcing-set-stpath}
    \shortincsssp{} can be solved in time $\bigoh(nm)$. 
    Moreover, when additionally given a shortest $s$-$t$ path $P$ in $G$, we can compute a minimum forcing set for $P$ in time $\bigoh(nm)$ as well.
\end{theorem}
It is not hard to see that, applying a standard trace-back technique, we can find a minimum forcing set for shortest $s$-$t$ paths within the same running time.
We would like to mention that our algorithm also works for undirected graphs since \cref{obs:shortest-path-dag} also holds for undirected graph $G$.

\subsection{Anti-forcing set}\label{subsec:exclude-stpath}
We next consider the complexity of \shortexcsssp{}.
In contrast to \shortincsssp, \shortexcsssp{} is \NP-complete even for undirected and unweighted graphs.
\begin{theorem}\label{thm:exclude-nph}
    \shortexcsssp{} is \NP-complete even for undirected and unweighted graphs.
    Moreover, \shortexcsssp{} is W[1]-hard with respect to $d$, where $d$ is the number of edges in a shortest path from $s$ to $t$ in the input undirected and unweighted graph.
\end{theorem}
We perform a parameterized reduction from \mcbiclique{}, which is defined as follows.
\begin{tcolorbox}
    \begin{description}
      \setlength{\itemsep}{0pt}
      \item[Problem:] \mcbiclique{}
      \item[Input:] A bipartite graph $G = (A \uplus B, E)$ and partitions $(A_1, \dots, A_k)$ of $A$ and $(B_1, \dots, B_k)$ of $B$ such that $|A_i| = |B_i| = n$ for every $1 \leq i \leq k$.
      \item[Goal:] Decide whether $G$ has a $k \times k$ \emph{multicolored biclique}, that is, a $k \times k$ biclique that contains exactly one vertex from each set $A_i$ and exactly one vertex from each set $B_i$.
    \end{description}
\end{tcolorbox}
This problem is NP-complete~\cite{GareyJ79} and W[1]-hard when parameterized by $k$~\cite[Exercise 13.3]{CyganFKLMPPS15} (see also \cite{Lin18}).

\paragraph{Construction}

Let $G = (A \uplus B, E)$ be a bipartite graph and let $(A_1, \dots, A_k)$ and $(B_1, \dots, B_k)$ be partitions of $A$ and $B$, respectively, such that $|A_i| = |B_i| = n$ for every $1 \leq i \leq k$.
From this instance, we construct a graph $H$ as follows (see \cref{fig:ex-sssp-w1hard} for an illustration):
\begin{enumerate}
    \item Initially, let $H = (A \cup B, \emptyset)$;
    \item for each $A_i$, add two vertices $s_i, t_i$ to $H$ and connect them to all vertices in $A_i$;
    \item for each $B_i$, add two vertices $s_{k + i}, t_{k + i}$ to $H$ and connect them to all vertices in $B_i$;
    \item for every $1 \leq i < 2k$, add an edge between $t_i$ and $s_{i + 1}$;
    \item\label{label:excsssp:w1-hard-wrt-d:enumerate-add-parallel-paths-between-non-edges} for each pair $(a, b) \in A \times B$ with $\{a, b\} \not\in E$, add $2k(n-1) + 2$ disjoint paths between $a$ and $b$ of length $3 (k+j-i)$, where $a \in A_i$ and $b \in B_j$.
\end{enumerate}

\begin{figure}
    \centering
    \includegraphics[width=\linewidth]{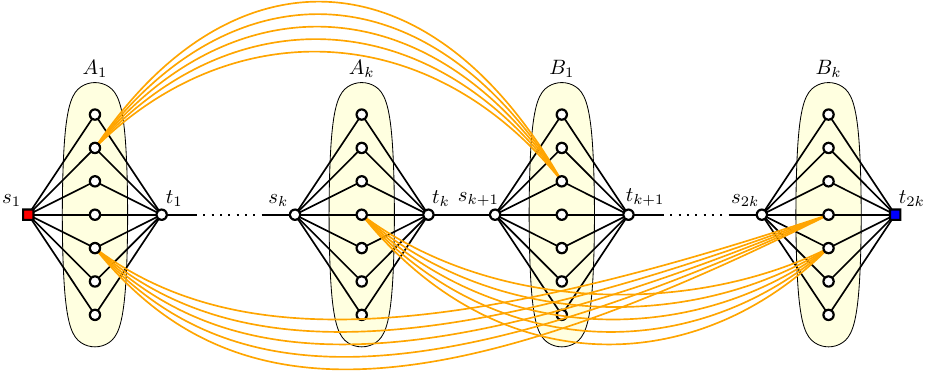}
    \caption{An illustration of the graph $H$ constructed from an instance of \mcbiclique{}.}
    \label{fig:ex-sssp-w1hard}
\end{figure}

We refer to paths added in Step~\ref{label:excsssp:w1-hard-wrt-d:enumerate-add-parallel-paths-between-non-edges} as \emph{bypasses}.
A path between $s_1$ and $t_{2k}$ in $H$ is said to be \emph{relevant} if, when it is oriented from $s_1$ to $t_{2k}$, it has neither $(u, s_i)$ nor $(t_i, u)$ for $u \in A \cup B$. 
Observe that every relevant path is a shortest path between $s_1$ and $t_{2k}$, and vice versa.
In the following we show that $G$ contains a $k \times k$ multicolored biclique if and only if there is an anti-forcing set $S \subseteq E(H)$ with $|S| \leq 2k (n-1)$ for shortest $s$-$t$ paths, which proves \cref{thm:exclude-nph}.
Note that the distance between $s_1$ and $t_{2k}$ is exactly $6k-1$, preserving the parameter.
    
\begin{lemma}
    There is a $k \times k$ multicolored biclique in $G$ if and only if there is an anti-forcing set $S \subseteq E(H)$ with $|S| \leq 2k (n-1)$.
\end{lemma}  
\begin{proof}
    Suppose that $G$ contains a $k \times k$ multicolored biclique $K \subseteq A \cup B$ such that $|K \cap A_i| = |K \cap B_i| = 1$ for $1 \le i \le k$.
    For $1 \leq i \leq k$, let $v_i \in K \cap A_i$ and $v_{k + i} \in K \cap B_i$.
    We then show that the edge set
    \begin{align*}
        S \coloneqq \left( \bigcup_{i = 1}^{k} \{\{s_i, u\}, \mid u \in A_i \setminus \{v_i\}\} \right) \cup \left( \bigcup_{i = 1}^{k} \{\{u, t_{k+i}\}, \mid u \in B_i \setminus \{v_{k+i}\}\} \right),
    \end{align*}
    whose size is $2k(n-1)$, is an anti-forcing set for shortest $s$-$t$ paths.
    Since $H-S$ contains a shortest path $P = (s_1, v_1, t_1, \dots, s_{2k}, v_{2k}, t_{2k})$, it suffices to show that no other paths exist.
    Suppose for the contrary that $H - S$ contains a relevant path $P'$ other than $P$.
    By the construction of $S$, $P'$ would have to contain a bypass.
    Let $(a, b) \in A_i \times B_j$ be the endpoints of this bypass.
    As $P'$ is relevant, $S$ does not contain edges $\{s_i, a\}$ and $\{b, t_{k+j}\}$.
    This hence implies that $a$ and $b$ are both in $K$, which indicates that no bypass between $a$ and $b$, a contradiction.

    Suppose that there is an anti-forcing set $S \subseteq E(H)$ of size at most $2k(n-1)$.
    Let $P$ be the unique relevant path between $s_1$ and $t_{2k}$ in $G - S$.
    Observe first that $P$ does not contain (an edge of) any bypass; otherwise $S$ must contain an edge in each of the other $2k(n-1) + 1$ bypasses with the same end vertices.
    Hence, $P$ visits all $s_i$'s and $t_i$'s, as well as a single vertex in each $A_i$ and $B_i$.
    Now we let $K_A = V(P) \cap A$ and $K_B = V(P) \cap B$.
    We then show that $K = K_A \cup K_B$ is a $k \times k$ biclique of $G$.
    Suppose for a contradiction that there exists $(a, b) \in K_A \times K_B$ such that $\{a, b\} \not\in E$.
    Then, $H$ contains $2k(n-1) + 2$ bypasses between $a$ and $b$.
    Since $|S| \le 2k(n-1)$, $H - S$ contains one of these bypasses, which contradicts the uniqueness of $P$.
    Hence, $K$ is a biclique of $G$.
\end{proof}
While \shortexcsssp{} is \NP-complete, we can find in polynomial time a minimum cost anti-forcing set for a given shortest $s$-$t$ path $P$ in $G$ by reducing it to the minimum multiway cut problem on directed acyclic graphs, which can be solved in polynomial time~\cite{Bentz07}.

\begin{theorem}\label{thm:specific-forcing-set-stpath}
    Let $G$ be an edge-weighted directed graph with $s, t \in V(G)$ and let $P$ be a shortest path from $s$ to $t$ in $G$.
    Then, a minimum anti-forcing set $S$ for $P$ can be computed in polynomial time.
\end{theorem}
\begin{proof}
	We reduce the problem to \mwcut{} on directed acyclic graphs.
	Let $H$ be a directed graph and let $T\subseteq V(H)$ be a terminal set. A \emph{multiway cut} for $T$ is a set $S \subseteq E(H)$ such that for any terminal pair $s, t \in T$, there is no path from $s$ to $t$ in $H-S$.
    Given a directed graph $H$ and $T \subseteq V(H)$, the task of \mwcut{} is to compute a minimum cardinality multiway cut of $(H, T)$.

	Due to \cref{obs:shortest-path-dag}, we can assume that the given graph $G$ is a directed acyclic graph and the paths from $s$ to $t$ are exactly the shortest $s$-$t$ paths in the original graph.
	We show that $S$ is an anti-forcing set for $s$-$t$ paths of $G$ if and only if $S$ is a multiway cut for $V(P)$ on $H \coloneqq G-E(P)$.

    Suppose that $S \subseteq E(G)$ is an anti-forcing set for $P$, that is, $P$ is the unique path from $s$ to $t$ with $S \cap E(P) = \varnothing$.
	Suppose that there is a path $P'$ from $u \in V(P)$ to $v \in V(P)$ in $H$.
    We then obtain an $s$-$t$ path $P''\; (\neq P)$ from $P$ by replacing the subpath between $u$ and $v$ with $P'$.
    Since $S$ is an anti-forcing set, $S$ contains at least one edge of $P''$, and thus, $E(P') \cap S \neq \varnothing$.
    Hence, $S$ is a multiway cut for $V(P)$ in $H$.

    Conversely, suppose that $S$ is a multiway cut for $V(P)$ in $H$.
    Since every path from $u$ to $v$ for any distinct $u, v \in V(P)$ in $H$ contains an edge in $S$, $P$ is the unique path from $s$ to $t$ in $G - S$, meaning that $S$ is an anti-forcing set for~$P$.
        
    It is known that \mwcut{} can be solved in polynomial time on directed acyclic graphs~\cite{Bentz07}.
	Thus, a minimum anti-forcing set $S$ for $P$ can be computed in polynomial time as well.
\end{proof}

We next discuss the parameterized complexity of \shortexcsssp{} when parameterized by the solution size~$k$.
Although we are unable to resolve this question fully, we can show that it is fixed-parameter tractable when additionally parameterizing the distance $d$ from $s$ to $t$ in the input graph.
We below give two algorithms with incomparable running times.

\begin{theorem}\label{thm:excsssp:d^k}
    \shortexcsssp{} can be solved in time $d^kn^{\bigoh(1)}$, where $d$ is the maximum number of edges in a shortest path from $s$ to $t$ in $G$.
\end{theorem}
\begin{proof}
    By~\cref{obs:shortest-path-dag}, we can assume that $G$ is acyclic and all the paths from $s$ to $t$ are shortest paths in $G$.
    Thus, it suffices to find a smallest edge set $S \subseteq E(G)$ such that $G - S$ has a unique path from $s$ to $t$, which corresponds to an anti-forcing set in the original graph.
    \begin{algorithm}[tb]
        \caption{A $\bigoh^*(d^k)$-time algorithm for \shortexcsssp{}.}
        \label{alg:anti-sssp1}
        \KwIn{A directed acyclic graph $G$, $s, t \in V(G)$, and a non-negative integer $k$.}
        \KwOut{Decide if $G$ has an edge set $S \subseteq E(G)$ such that $G - S$ has a unique path from $s$ to $t$.}
        \SetKwProg{Fn}{Function}{:}{}
        \SetKwFunction{FnBranch}{branching}

        \Fn{\FnBranch{$G, k$}}{
            Let $P$ be an arbitrary path from $s$ to $t$ in $G$\;
            Let $S$ be a minimum anti-forcing set for $P$ computed by the algorithm in \cref{thm:specific-forcing-set-stpath}\;
            \If{$|S| \le k$}{
                \Return \texttt{true}\;
            }
            \If{$k > 0$}{
                \For{$e \in E(P)$} {
                    \If{{\FnBranch{$G - e, k - 1$}}}{
                        \Return \texttt{true}\;
                    }
                }
            }
            \Return \texttt{false}\;
        }
    \end{algorithm}
    The branching algorithm \FnBranch{$G, k$} described in \Cref{alg:anti-sssp1} correctly decides whether $G$ has such a set of size at most~$k$.
    
    Let $(G, k)$ be the input of this branching algorithm.
    Suppose that $G$ has an anti-forcing set $S^*$ of size at most $k$.
    Let $P$ be an arbitrary path from $s$ to $t$ in $G$.
    If $S^*$ is disjoint from $E(P)$, we can find a smallest anti-forcing set $S$ in polynomial time using \cref{thm:specific-forcing-set-stpath}.
    Note that $S = \emptyset $ when $P$ is already unique in $G$.
    Otherwise, $P$ contains at least one edge of $S^*$.
    In this case, the budget~$k$ must be positive, and we then branch all edges $e$ in $P$ and seek a smallest anti-forcing set of $(G - e, k - 1)$.
    Thus, this branching algorithm correctly decides whether $G$ has an anti-forcing set of size at most~$k$.

    The running time of \Cref{alg:anti-sssp1} is readily bounded by $d^kn^{\bigoh(1)}$ as every path from $s$ to $t$ in $G$ has at most $d$ edges. 
\end{proof}

\begin{theorem}\label{thm:excsssp:k^d}
    \shortexcsssp{} can be solved in time $(k + 1)^dn^{\bigoh(1)}$, where $d$ is the maximum number of edges in a shortest path from $s$ to $t$ in $G$.
\end{theorem}
To prove \cref{thm:excsssp:k^d}, we need several lemmas.
Similarly to the previous proof, we assume that $G$ is acyclic, and all the paths from $s$ to $t$ are shortest in $G$.
Let $S^*$ be a smallest anti-forcing set in $G$ and let $P^*$ be the unique path in $G - S^*$.
Let $\lambda_G(u, v)$ be the maximum number of edge-disjoint paths from $u$ to $v$ in $G$.
     
\begin{lemma}\label{lem:large-cut}
    Let $u$ and $v$ be distinct vertices of $P^*$ with $\dist(s, u) < \dist(s, v)$.
    If $\lambda_G(u, v) > k + 1$, then $G$ has no anti-forcing set of size at most $k$.
\end{lemma}
\begin{proof}
     Since $P^*$ is the unique path from $s$ to $t$ in $G - S^*$, the subpath from $u$ to $v$ is also unique.
     This means that $\lambda_{G - S^*}(u, v) = 1$.
    As deleting a single edge decreases $\lambda_{G}(u, v)$ by at most~1, $S^*$ must contain at least $k + 1$ edges. 
\end{proof}
    
Our algorithm constructs a set of paths $\mathcal P_{st}$ that contains the unique path $P^*$ from $s$ to $t$ in $G - S^*$ for every anti-forcing set $S^*$ of size at most~$k$.
Moreover, $\mathcal P_{st}$ contains no more than $(k + 1)^d$ paths.
This proves \cref{thm:excsssp:k^d} due to \cref{thm:specific-forcing-set-stpath}.
Let $(x, y)$ be an edge of $G$ and let $\mathcal P_{x}$ and $\mathcal P_{y}$ be sets of paths ending at $x$ and starting at $y$, respectively.
We denote by $\mathcal P_x \otimes \mathcal P_y$ the set of paths obtained by concatenating every pair of paths in $\mathcal P_x$ and in $\mathcal P_y$.
The algorithm for constructing $\mathcal P_{st}$, shown in \Cref{alg:anti-sssp2}, computes a minimum $s$--$t$ cut in $G$ and, for each edge $(x, y)$ in the cutset, recursively constructs sets of paths $\mathcal P_{sx}$ and $\mathcal P_{yt}$.
Then, we add $\mathcal P_{sx} \otimes \mathcal P_{yt}$ to $\mathcal P_{st}$.
    \begin{algorithm}[tb]
        \caption{A divide-and-conquer algorithm for constructing $\mathcal P_{st}$.}
        \label{alg:anti-sssp2}
        \SetKwProg{Fn}{Function}{:}{}
        \SetKwFunction{FnConstruct}{construct}

        \Fn{\FnConstruct{$G, k, u, v$}}{
            \If{$u = v$}{
                \Return A singleton path containing $u$\;
            }
            Let $C \subseteq E(G)$ be the cutset of a minimum cut separating $u$ and $v$ in $G$\; \label{line:sssp2:cut}
            \If{$|C| > k + 1$}{
                \Return $\emptyset$\;
            }
            $\mathcal P_{uv} \leftarrow \emptyset$\;
            \For{$(x, y) \in C$}{
                Let $\mathcal P_{ux} \leftarrow$ \FnConstruct($G, k, u, x$)\;
                Let $\mathcal P_{yv} \leftarrow $ \FnConstruct($G, k, y, v$)\;
                $\mathcal P_{uv} \leftarrow \mathcal P_{uv} \cup (\mathcal P_{ux} \otimes \mathcal P_{yv})$\;
            }
            \Return $\mathcal P_{uv}$;
        }
    \end{algorithm}

Let $\mathcal P_{st}$ be the set of paths obtained by calling \FnConstruct{$G, k, s, t$}.
For $u, v \in V(G)$, let $d(u, v)$ be the maximum number of edges in a shortest path from $u$ to $v$ in $G$.
\begin{lemma}\label{lem:k^d:optimality}
    Let $S^*$ be an arbitrary anti-forcing set of size at most $k$ and let $P^*$ be the unique path from $s$ to $t$ in $G - S^*$.
    Then, $P^* \in \mathcal P_{st}$.
\end{lemma}
\begin{proof}
    It suffices to show the following claim: if \FnConstruct{$G, k, u, v$} is called for some $u, v \in V(P^*)$, then the returned set $\mathcal P_{uv}$ always contains the subpath $P^*_{uv}$ of $P^*$ from $u$ to $v$.
    We prove it by induction on $d(u, v)$.
    
    The base case, $u = v$, is clear.
    Suppose now that the claim holds for all pairs of vertices $a,b$ with $d(a,b) < d(u,v)$.
    By~\cref{lem:large-cut}, $\lambda_G(u, v) \le k + 1$ holds.
    Let $C$ be the cutset computed in line~\ref{line:sssp2:cut}.
    As there is a path from $u$ to $v$ in $G$, $C$ contains at least one edge in $P$.
    For such an edge $(x, y) \in C \cap E(P)$, the algorithm constructs $\mathcal P_{ux}$ and $\mathcal P_{yv}$.
    As $d(u, x) + d(y, v) < d(u, v)$, it holds that $P^* \in \mathcal P_{ux}$ and $P^*_{yv} \in \mathcal P_{yv}$ due to the induction hypothesis.
    This implies that $P^*_{uv} \in \mathcal P_{ux} \otimes \mathcal P_{yv} \subseteq \mathcal P_{uv}$.
\end{proof} 

\begin{lemma}\label{lem:k^d:running time}
    \Cref{alg:anti-sssp2} returns a set of paths $\mathcal P_{st}$ in $(k + 1)^dn^{\bigoh(1)}$ time.
    In particular, $\mathcal P_{st}$ contains at most~$(k + 1)^d$ paths.
\end{lemma}
\begin{proof}
    It suffices to show the following claim: For vertices $u$ and $v$, the set $\mathcal P_{uv}$, obtained by calling \FnConstruct{$G, k, u, v$}, contains at most $(k + 1)^{d(u, v)}$ edges.
    We prove it by induction on $d(u, v)$.
    
    The base case, $u = v$, is clear.
    Thus, suppose that the claim holds for all pairs of vertices $a,b$ with $d(a,b) < d(u,v)$.
    We can assume that the cutset $C$ computed at line~\ref{line:sssp2:cut} contains at most $k + 1$ edges; otherwise the claim is trivial.
    For each $(x, y) \in C$, we have $d(u, x) + d(y, v) < d(u, v)$.
    Applying the induction hypothesis, we have $|\mathcal P_{ux}| \le (k + 1)^{d(u, x)}$ and $|\mathcal P_{yv}| \le (k + 1)^{d(y, v)}$.
    Thus, we have
    \begin{align*}
        |\mathcal P_{uv}| &= |\bigcup_{(x, y) \in C}\mathcal P_{ux} \otimes \mathcal P_{yv}|\\
        &\le \sum_{(x, y) \in C}|\mathcal P_{ux} \otimes \mathcal P_{yv}|\\
        &\le \sum_{(x, y) \in C} (k + 1)^{d(u, x) + d(y, v)}\\
        &\le (k + 1) \cdot (k + 1)^{d(u, v) - 1}\\
        &\le (k + 1)^{d(u, v)}.
    \end{align*}
    As we can compute a minimum $s$-$t$ cut in polynomial time, the claim for the running time immediately follows.
\end{proof}

We are ready to describe our main algorithm for \cref{thm:excsssp:k^d}.
We first construct a set $\mathcal P_{st}$ of paths from $s$ to $t$ in $G$ using \Cref{alg:anti-sssp2}.
Then, for each $P \in \mathcal P_{st}$, we compute in polynomial time a smallest set $S \subseteq E(G)$ such that $P$ is the unique path from $s$ to $t$ in $G - S$ using \cref{thm:specific-forcing-set-stpath}.
If there is a path $P \in \mathcal P_{st}$ such that the corresponding set $S$ has at most~$k$ edges, we can indeed conclude that $G$ has an anti-forcing set of size at most~$k$; otherwise $G$ has no such a set due to \cref{lem:k^d:optimality}.
By~\cref{lem:k^d:running time}, the entire algorithm runs in time $(k + 1)^dn^{\bigoh(1)}$.

Finally, we observe that the problem can be solved in linear time on graphs of bounded treewidth.
Since it is not necessary to discuss what treewidth is in the proof, we omit its definition in this paper (see e.g.~\cite[Chapter 7]{CyganFKLMPPS15}).
We observe that \shortexcsssp{} is reduced to one for evaluating a monadic second-order ($\MSO_2$) formula, and thus, by Courcelle's theorem~\cite{CourcelleM93,ArnborgLS91} and Bodlaender's algorithm~\cite{Bodlaender96},
it can be solved in linear time on graphs of bounded treewidth.
\begin{theorem}\label{thm:treewidth-exclude-stpath}
	Let $G$ be a directed graph with two specified vertices $s$ and $t$ such that the underlying graph of $G$ has bounded treewidth.
	Then, a minimum anti-forcing set $S$ for shortest $s$-$t$ paths on $G$ can be computed in linear time.
\end{theorem}
\begin{proof}
Due to \cref{obs:shortest-path-dag}, it is sufficient to compute, for directed acyclic graph $G'$, a minimum edge set $S$ such that $G'-S$ has a unique $s$-$t$ path.
We express this condition by an $\MSO_2$-formula over directed graphs with two constant vertices $s$ and $t$.
In $\MSO_2$ logic, we can quantify edges, vertices, vertex sets, and edge sets, and we can use the standard (logical) connectives `$=$', `$\land$', `$\lnot$', and so on~(for a more formal definition, see e.g., ~\cite{CyganFKLMPPS15}).
For edge variable $e=(u,v)$, we use $\lhead(e)$ and $\ltail(e)$ that are functions returning $v$ and $u$, respectively\footnote{Formally, we should encode these functions as a relation like $\lhead(e,v)$ and $\ltail(e,u)$. However, for simplicity, we use the function notation here, and it is easy to translate to the formal relation-based notation.}.

We use $\exists! X\; \phi(X)$ as an abbreviation for $\exists X\; \forall Y\; [X=Y \leftrightarrow \phi(Y)]$.
Note that $\exists! X\; \phi(X)$ is uniqueness quantification, that is, it can be interpreted as meaning ``there is exactly one $X$ such that $\phi(X)$''.
Now, we consider the following $\MSO_2$ formula\footnote{We use abbreviated notations such as $\subseteq$, $\cap$, $=\varnothing$, $\forall x \in X \phi(X)$, but these are interpreted in the standard way and can be defined in $\MSO_2$ logic.}
\begin{align*}
\phi(X) &\equiv X \subseteq E \land \exists! P \;(P \subseteq E \land P \cap X = \varnothing \land \texttt{Path}_{s,t}(P)).
\end{align*}
Here, $\texttt{Path}_{s,t}(P)$ is a formula that expresses that a directed edge set $P$ is an $s$-$t$ path.
It is known that the formula $\texttt{Path}_{s,t}(P)$ can be defined by an $\MSO_2$-formula~(see e.g., ~\cite{Courcelle97}), which is given below for completeness.
\begin{gather*}
\texttt{Path}_{s,t}(P)  \equiv \texttt{QPath}_{s,t}(P) \land \forall P' \subseteq P\; (P = P' \lor \lnot \texttt{QPath}_{s,t}(P')),\\
    \texttt{QPath}_{s,t}(P)  \equiv  \begin{pmatrix}
    \exists! e\; (e \in P \land s = \ltail(e))\land \exists! f\; (f \in P \land t = \lhead(f)) \land {}\\
    \forall e \in P \; \lnot (s = \lhead(e) \lor t =\ltail(e)) \land {}\\
    \forall e \in P \; \begin{pmatrix}
    (\lhead(e) \neq t) \leftrightarrow \exists! f \;(f \in P \land \lhead(e) = \ltail(f))  & \land\\
    (\ltail(e) \neq s) \leftrightarrow \exists! f \;(f \in P \land \ltail(e) = \lhead(f)) 
    \end{pmatrix}
\end{pmatrix}.
\end{gather*}
Here, the condition $\texttt{QPath}_{s,t}(P)$ ensures that $P$ contains an $s$-$t$ path, but $P$ may additionally contain vertex-disjoint cycles.
Therefore, the minimality condition $\forall P' \subseteq P\; (P = P' \lor \lnot \texttt{QPath}_{s,t}(P'))$ guarantees that $P$ is exactly a directed $s$-$t$ path.

It is known that a set $S$ that satisfies the property $\phi(S)$ defined by a constant size $\MSO_2$-formula $\phi$ and minimizes its cardinality $|S|$ can be found in linear time on directed graphs such that the underlying graph has bounded treewidth~\cite{Bodlaender96,CourcelleM93,ArnborgLS91}.
Since the underlying graph $G''$ of $G'$ is a subgraph of $G$, $G''$ has bounded treewidth.
Thus, the claim holds.
\end{proof}
Note that we use the cardinality of an edge-set variable $X$ as an objective in the proof of \cref{thm:treewidth-exclude-stpath}, and thus we cannot directly extend the result for bounded clique-width since $\MSO_2$ model checking is hard even for cliques~\cite{CourcelleMR00}.

\printbibliography
\end{document}